\documentclass[10pt,journal]{IEEEtran}
\usepackage{cite}
\usepackage{amsmath,amssymb}
\usepackage{graphicx}
\usepackage{booktabs}
\usepackage[table]{xcolor}
\usepackage{hyperref}
\hypersetup{hidelinks}
\newcommand{\tablestyle}[2]{\setlength{\tabcolsep}{#1}\renewcommand{\arraystretch}{#2}}
\newcommand{\tablefontsize}{\fontsize{7.2pt}{8.4pt}\selectfont}
\newcommand\headspace{\hspace{.2em}}
\newcommand{\err}[1]{{\color{black!55}\ensuremath{\pm #1}}}
\newtheorem{proposition}{Proposition}

\title{Condition-Wise Sinkhorn Drifting for One-Shot Learned Channel Simulation}

\author{Rick Fritschek,~\IEEEmembership{Member,~IEEE,} and Rafael F. Schaefer,~\IEEEmembership{Senior Member,~IEEE}%
\thanks{Rick Fritschek and Rafael F. Schaefer are with the Chair of Information Theory and Machine Learning, Technische Universit{\"a}t Dresden, Germany.}%
\thanks{This work was supported in part by the German Federal Ministry of Research, Technology and Space (BMFTR) through the Transfer Hub \emph{6G-life} under Grant 16KIS2413K and in part by the German Research Foundation (DFG, Deutsche Forschungsgemeinschaft) as part of Germany's Excellence Strategy--EXC 2050/2--Project ID 390696704--Cluster of Excellence \emph{``Centre for Tactile Internet with Human-in-the-Loop'' (CeTI)} of TUD Dresden University of Technology.}}

\begin{document}
\maketitle

\begin{abstract}
Learned communication systems may evaluate stochastic channel surrogates millions of times inside differentiable training loops, making diffusion-style reverse sampling expensive.
This paper proposes condition-wise Sinkhorn drifting, a one-shot channel surrogate that preserves the transmitted symbol and transports only the conditional output laws \(p(y\mid x)\).
We formulate a conditional Sinkhorn objective over repeated outputs at the same transmitted symbol and train the generator with finite-sample barycentric velocities followed by detached particle regression.
Experiments on additive white Gaussian noise (AWGN), Rayleigh fading, solid-state power amplifier (SSPA) nonlinearity, and a compact tapped-delay-line (TDL) channel compare direct drifting, joint Sinkhorn drifting, condition-wise Sinkhorn drifting, conditional denoising diffusion probabilistic modeling (DDPM), denoising diffusion implicit modeling (DDIM), and Wasserstein generative adversarial network (WGAN) references.
Within the evaluated one-shot drifting-family variants, condition-wise Sinkhorn is strongest under conditional diagnostics and symbolic-coding checks, while diffusion remains strongest on the hardest downstream symbol-error-rate (SER) curves.
The resulting operating point is a condition-preserving one-shot simulator for settings where repeated channel calls make diffusion-style sampling too costly.
\end{abstract}

\begin{IEEEkeywords}
learned channel simulation, conditional generative models, Sinkhorn divergence, diffusion models, inference latency
\end{IEEEkeywords}

\section{Introduction}
Accurate channel models are a basic requirement for modern communication system design.
They are particularly important in neural and end-to-end learned systems, where the channel appears repeatedly inside training loops, decoder adaptation procedures, and differentiable autoencoder pipelines.
When a clean analytical channel model is unavailable, too restrictive, or poorly matched to measured data, generative channel surrogates become attractive because they can learn the full conditional distribution directly from samples.

This problem also involves more than distributional accuracy.
In many learning pipelines, the channel model is evaluated millions of times.
As a result, sample quality, conditioning accuracy, computational cost, and differentiability all matter simultaneously.
The practical question is whether a generative model can match the channel distribution at a cost compatible with large-scale training and deployment workflows.

Recent work on diffusion-based channel modeling has shown that conditional diffusion models can generate high-quality channel samples across standard benchmarks such as additive white Gaussian noise (AWGN), Rayleigh fading, and nonlinear amplifier channels \cite{paper2309}.
Follow-up work has further studied robustness and speed-quality tradeoffs for diffusion-based channel generation \cite{kimicc2024}, extended diffusion-based channel synthesis to high-dimensional user-specific channels \cite{lee2024highdimchannels}, and proposed diffusion-driven digital-twin-style generation of statistical channel state information from sensing and location information \cite{gong2025dtoc}.
More broadly, diffusion methods have also begun to influence adjacent wireless tasks such as high-dimensional channel estimation \cite{zhou2025channelestimation}, low-complexity multiple-input multiple-output (MIMO) channel estimation \cite{fesl2024mimoestimation,chen2025mimovi}, joint channel estimation and data detection \cite{zilberstein2024joint}, and fast time-varying MIMO orthogonal frequency-division multiplexing (OFDM) estimation \cite{liu2025mimoofdm}.
This line of work is compelling because diffusion models are flexible, robust, and often more stable than earlier adversarial approaches.
This improvement comes with a practical system cost.
Inference requires an iterative reverse denoising procedure, and the latency scales with the number of diffusion steps.
This tradeoff is manageable when the channel generator is called occasionally.
It becomes more problematic when the generator is embedded in inner loops of iterative training or repeated Monte Carlo evaluations.
This latency bottleneck has motivated several acceleration strategies for the denoising diffusion probabilistic model (DDPM) \cite{ddpm}, including skipped or reduced-step sampling as in denoising diffusion implicit modeling (DDIM), as well as one-step or few-step alternatives such as consistency models and distillation-based approaches that compress multi-step samplers into substantially shorter generation paths \cite{ddim,song2023consistency,yin2024onestep}.
This motivates the study of one-shot conditional generators for learned channel simulation.
Among recently proposed generative paradigms, drifting models \cite{drifting} are particularly interesting because they learn the transport from a simple latent source to the target data distribution directly in the network parameters.
Intuitively, diffusion spends computation at test time by progressively transforming noise into data through a reverse-time chain.
Drifting spends computation during training.
Generated samples are repeatedly nudged toward the data manifold by a drift field, and the neural generator is updated so that future samples move closer to the desired distribution in a single forward pass.
Once trained, sampling is therefore naturally one-shot.
Here, one-shot means one neural generator evaluation per channel sample, with the usual latent draw and conditioning input but no reverse denoising chain or iterative particle correction at inference.
This perspective is also consistent with recent geometric analyses of generative dynamics that relate diffusion-like behavior to autonomous vector-field structure and marginalized noise geometry \cite{geometryofnoise}.
Recent work on one-step generation via Wasserstein gradient flows (W-Flow) \cite{han2026wflow} sharpens this viewpoint by replacing heuristic drifting fields with Sinkhorn-divergence optimal-transport velocities.
This is directly relevant to channel simulation, because the quality and stability of the drift field is the main design question once inference is constrained to one generator evaluation.

The distinction between ordinary generation and channel simulation is central.
In ordinary generative modeling, samples are transported in the data space until the model distribution matches the target distribution.
In channel simulation, the transmitted symbol is fixed side information.
The relevant target is the family of conditional laws \(\{p(\cdot\mid x)\}_{x\sim\mu}\).
A generator can match the global output cloud while mixing samples from different transmitted symbols, which produces an incorrect channel surrogate for communication-system training.
This motivates a condition-wise transport formulation in which couplings preserve \(x\) and move only the output component \(y\).
The proposed condition-wise Sinkhorn drifting implements this idea as a one-shot neural channel generator.
It uses drifting's detached-target regression and the Sinkhorn/W-Flow barycentric velocity as building blocks.
The channel-specific contribution is to constrain the transport problem to repeated outputs at the same transmitted symbol, so that the learned surrogate moves samples inside \(p(y\mid x)\) rather than across different transmitted symbols.

The emphasis of this study is on generator-level distributional accuracy, computational tradeoffs, and downstream channel-implant checks.
A full communication-system design study over all modulation and coding regimes is outside the present scope.
The condition-wise Sinkhorn estimator also assumes simulator-style access to repeated outputs at the same transmitted symbol.
Measured datasets with only one observation per condition require local conditional approximations, which we formulate only as a limitation and future extension.

The main contributions are summarized as follows.
\begin{itemize}
\item We formulate condition-wise Sinkhorn drifting for learned channel simulation by integrating Sinkhorn divergences over the fixed-input conditional laws \(p(\cdot\mid x)\).
\item We connect this objective to a conditional Wasserstein-flow velocity, then state how the practical method differs from the population flow through finite-sample barycentric estimation and projected detached-target neural training.
\item We implement direct-output conditional drifting and compare joint/global Sinkhorn transport with condition-wise Sinkhorn transport under a matched one-shot generator.
\item We compare the proposed one-shot surrogate against conditional DDPM, conditional DDIM, and conditional Wasserstein generative adversarial network (WGAN) reference models across AWGN, Rayleigh fading, solid-state power amplifier (SSPA), and compact tapped-delay-line (TDL) channels using generator-level metrics, downstream coding checks, and latency measurements.
\item We evaluate channel surrogates with direct sliced Wasserstein distance (SWD), anchor-conditioned moment metrics, downstream symbol-error rate (SER), bit error rate (BER), and timing measurements, showing that communication usefulness is better reflected by conditional and downstream metrics than by global SWD alone.
\end{itemize}
Source code will be made available through the first author's GitHub repository upon publication\footnote{\url{https://github.com/Fritschek/conditional_drifting_models}.}.

\section{Problem Formulation and Drifting Background}
\subsection{What drifting changes conceptually}
Because drifting is very recent and remains less familiar than diffusion, we briefly summarize its generative viewpoint.
Following \cite{drifting}, let \(f_\theta\) map a latent variable \(z \sim p_z\) to a sample \(u = f_\theta(z)\), and let
\begin{equation}
q_\theta = (f_\theta)_\# p_z
\end{equation}
denote the pushforward distribution induced by the generator.
The training objective uses a \emph{drifting field} that tells us how samples drawn from the current model distribution should move so that the model distribution approaches the target distribution \(p\).

At a high level, the desired property is an equilibrium condition.
If the model already matches the target distribution, then the drift should vanish.
Using \(V_{p,q}(u)\) to denote the field evaluated at a sample \(u\) under target distribution \(p\) and model distribution \(q\), the fixed-point intuition can be summarized as
\begin{equation}
q = p
\quad \Longrightarrow \quad
V_{p,q}(u) = 0,
\end{equation}
which is the sense in which matched distributions form an equilibrium of the training dynamics \cite{drifting}.

This makes the training logic very different from diffusion.
In diffusion, one learns a denoiser and then performs iterative refinement at inference time.
In drifting, one instead constructs improved targets during training by moving current samples along the field \(V_{p,q}\), and then regresses the generator onto those improved targets.
As a result, the iterative correction happens during optimization.

The one-step view can be written as follows.
Starting from a generated sample \(u = f_\theta(z)\), define a drifted target
\begin{equation}
\tilde{u}
=
\mathrm{stopgrad}\!\left(
u + \eta V_{p,q_\theta}(u)
\right),
\end{equation}
and optimize the generator by
\begin{equation}
\mathcal{L}_{\text{drift}}
=
\mathbb{E}_{z \sim p_z}
\left[
\left\|
f_\theta(z) - \tilde{u}
\right\|_2^2
\right].
\end{equation}
The stop-gradient is essential here because it makes the drifted sample behave like a fixed target for the current optimization step.
Repeated training steps therefore evolve the pushforward \(q_\theta\) toward the target distribution.
This is also the cleanest way to see why drifting is naturally one-shot at inference.
Once the generator has learned to emit samples near equilibrium, sampling only requires drawing \(z\) and evaluating \(f_\theta(z)\), without a reverse-time chain.

The W-Flow formulation \cite{han2026wflow} connects this one-shot training principle to entropic optimal transport.
Let \(\operatorname{OT}_{\varepsilon}(q,p)\) denote the entropic optimal-transport cost between a generated law \(q\) and target law \(p\), computed with Sinkhorn iterations \cite{cuturi2013sinkhorn}.
The debiased Sinkhorn divergence is
\begin{equation}
S_{\varepsilon}(q,p)
=
\operatorname{OT}_{\varepsilon}(q,p)
-\frac{1}{2}\operatorname{OT}_{\varepsilon}(q,q)
-\frac{1}{2}\operatorname{OT}_{\varepsilon}(p,p),
\end{equation}
which removes the entropic self-bias and interpolates between optimal-transport and kernel-like behavior as \(\varepsilon\) changes \cite{feydy2019sinkhorn}.
For an empirical generated particle \(u\), the corresponding W-Flow velocity can be written in barycentric-projection form as
\begin{equation}
V_{\varepsilon}(u)
=
T^{\varepsilon}_{q,p}(u)
-
T^{\varepsilon}_{q,q}(u),
\label{eq:wflow_velocity}
\end{equation}
where \(T^{\varepsilon}_{q,p}\) maps a generated sample to the barycenter of target samples under the entropic optimal-transport coupling, and \(T^{\varepsilon}_{q,q}\) is the generated self-transport projection.
This has the same attraction--repulsion structure as the practical drifting rule.
The weights come from a globally mass-constrained coupling that replaces independently row-normalized kernels.

\subsection{Conditional channel simulation}
We now specialize the generic drifting view to learned conditional channel simulation.
Let \(x \in \mathbb{R}^n\) denote the transmitted channel input and \(y \in \mathbb{R}^n\) the channel output.
The objective is to learn a conditional generator that reproduces the channel law given \(x\).
In this setting, conditioning enters through the transmitted symbol, while stochasticity is carried by a latent variable \(z \sim \mathcal{N}(0, I)\).

The experiments use a direct channel-output formulation.
The generator approximates \(p(y \mid x)\) and predicts
\begin{equation}
\hat{y} = g_\theta(x,z).
\end{equation}
This choice aligns the drifting generator with the direct-output diffusion and WGAN baselines and with the downstream channel-implant evaluation, where the learned surrogate is used as a simulator for \(y\) given \(x\).
Residual-output variants were examined during development and left out of the final benchmark after weaker direct-output and downstream channel-implant performance.

\subsection{Conditional drifting model}
  We use a practical conditional adaptation inspired by \cite{drifting}.
  To avoid confusion with the communication input symbol \(x\), we use \(u\) above for the generic
  sample variable in the abstract formulation and reserve \(x\) below for the transmitted symbol.
  In our implementation, the generator predicts \(\hat y = g_\theta(x,z)\), where
  \begin{equation}
  z \sim \mathcal{N}(0, I).
  \end{equation}
  The network is implemented as a two-hidden-layer multilayer perceptron (MLP) with hidden
  width \(128\), latent dimension \(16\), and sigmoid linear unit (SiLU) activations.
\begin{equation}
[x, z]
\rightarrow
\mathrm{Linear}
\rightarrow
\mathrm{SiLU}
\rightarrow
\mathrm{Linear}
\rightarrow
\mathrm{SiLU}
\rightarrow
\mathrm{Linear}.
\end{equation}
The training target is formed from an attraction--repulsion drift field in direct output space. This is a concrete channel-specific instantiation of the generic field \(V_{p,q}\).
Given generated samples \(G = \{g_i\}\) and positive samples \(P = \{p_j\}\), we compute
\begin{equation}
V(g_i)
=
\left(
\sum_j w^{(P)}_{ij} p_j - g_i
\right)
-
\lambda
\left(
\sum_k w^{(G)}_{ik} g_k - g_i
\right),
\end{equation}
where \(w^{(P)}\) and \(w^{(G)}\) are row-normalized radial basis function kernel weights.
The first term attracts each generated sample toward a local barycenter of true output samples, and the second term repels it from nearby generated samples to reduce collapse.
This yields a detached target
\begin{equation}
\tilde{g}_i
=
\mathrm{stopgrad}\!\left(g_i + \eta V(g_i)\right),
\end{equation}
and the generator is optimized by
\begin{equation}
\mathcal{L}_{\text{drift}}
=
\frac{1}{|G|}
\sum_i
\left\|
g_i - \tilde{g}_i
\right\|_2^2.
\end{equation}
The generic update \(u \mapsto u + \eta V_{p,q}(u)\) from the drifting formulation is realized here with a kernel-based drift field estimated directly from minibatches of true and generated output samples.
This empirical field acts as a kernel particle interaction rule.
The attraction term pulls generated samples toward the target sample cloud, while the repulsion term prevents generated particles from collapsing onto a few modes.
In this sense, the method is related in spirit to kernel-based distribution shaping and maximum mean discrepancy matching.
Here it is used as a practical training rule, without deriving a formal discrepancy minimization objective.

In the simulations, the drift scale is \(\eta=1.0\), the repulsive weight is \(\lambda=1.0\), the minimum kernel bandwidth is \(0.2\), and the maximum drift norm is \(2.0\).
These settings were chosen because they produced stable training behavior and materially improved the early AWGN drifting baselines.

The direct-drifting kernel field serves as the plain one-shot drifting baseline in the experiments.
Its quality still depends on bandwidth selection and on how minibatch interactions approximate the conditional channel geometry.
To reduce this channel-specific kernel-design burden, the main drift-field ablation replaces the row-normalized kernel interaction by Sinkhorn/W-Flow barycentric velocities while preserving the same detached-target training loop and the same one-shot inference path.

\subsection{Condition-wise Sinkhorn drifting}
The W-Flow objective in \cite{han2026wflow} is unconditional.
Both \(p\) and \(q_\theta\) are distributions on the generated sample space.
For channel simulation, the generated and true joint laws share the same transmitted-symbol marginal.
Writing \(\mu=P_X\) for the input marginal and \(p_x=P_{Y\mid X=x}\) for the conditional output law, the target and generated joint laws disintegrate as
\begin{equation}
p(dx,dy)=\mu(dx)p_x(dy),
\qquad
q_\theta(dx,dy)=\mu(dx)q_{\theta,x}(dy).
\end{equation}
The condition \(x\) is fixed side information.
Transporting in the full joint space \((x,y)\) therefore solves a different problem from matching the channel law \(p(y\mid x)\), because it allows the geometry of the fixed condition variable to influence the coupling while the update can only move the output component.

We instead define the conditional Sinkhorn functional
\begin{equation}
\mathcal{S}^{\mathrm{cond}}_{\varepsilon}(q_\theta,p)
=
\int
S_{\varepsilon}\!\left(q_{\theta,x},p_x\right)
\mu(dx),
\label{eq:conditional_sinkhorn}
\end{equation}
where the Sinkhorn divergence is computed separately over output samples for each fixed transmitted symbol.
For brevity, we call the output cloud at a fixed transmitted symbol an output fiber.
In measure-theoretic terms, \(p_x=p(\cdot\mid x)\) is the conditional measure over that fiber.
Throughout the paper we refer to the resulting estimator as condition-wise Sinkhorn drifting.
Equivalently, admissible couplings must preserve the condition.
\begin{equation}
\pi(dx,dy,d\bar y)
=
\mu(dx)\pi_x(dy,d\bar y),
\qquad
\pi_x \in \Pi(q_{\theta,x},p_x).
\end{equation}
The Wasserstein gradient flow of \eqref{eq:conditional_sinkhorn} has zero \(x\)-velocity and evolves only inside each fixed-input conditional law.
\begin{equation}
\partial_t q_t(x,y)
+
\nabla_y\!\cdot
\left(
q_t(x,y)v_t(x,y)
\right)
=0.
\end{equation}
Applying the W-Flow velocity condition-wise gives
\begin{equation}
v_t(x,y)
=
T^{\varepsilon}_{q_{t,x},p_x}(y)
-
T^{\varepsilon}_{q_{t,x},q_{t,x}}(y).
\label{eq:condition_sinkhorn_velocity}
\end{equation}

For channel simulation, the key consequence is the following equilibrium statement.
\begin{proposition}[Condition-wise Sinkhorn equilibrium]
Under the standard Sinkhorn-divergence assumptions of \cite{feydy2019sinkhorn}, the conditional objective \(\mathcal{S}^{\mathrm{cond}}_{\varepsilon}\) is nonnegative and is zero exactly when
\begin{equation}
q_{t,x}=p_x
\quad
\text{for } \mu\text{-a.e. }x.
\end{equation}
Along the corresponding population W-Flow dynamics,
\begin{equation}
\frac{d}{dt}\mathcal{S}^{\mathrm{cond}}_{\varepsilon}(q_t,p)
=
-
\int
\|v_t(x,\cdot)\|_{L^2(q_{t,x})}^2
\mu(dx)
\leq 0.
\end{equation}
Thus the equilibrium of the condition-wise flow is equality of the full family of channel laws \(p(\cdot\mid x)\).
\end{proposition}
\emph{Proof:}
The zero-set statement follows by applying the nonnegativity and unique-zero property of the Sinkhorn divergence to each fixed \(x\) and integrating over \(\mu\).
The derivative identity is the same fixed-condition W-Flow identity, again integrated over \(\mu\).
\hfill\(\square\)

The proposition gives the population conditional optimal-transport picture needed for channel simulation: equality is enforced for almost every fixed input.
Matching a joint sample cloud under an arbitrary joint cost does not enforce this condition-wise equality.

The implemented algorithm approximates this population object in two ways.
First, it estimates the barycentric velocity from finite minibatch Sinkhorn couplings.
Second, it projects the explicit particle update back into the neural generator class through detached-target regression.
The resulting update differs from exact gradient descent on \(\mathcal{S}^{\mathrm{cond}}_{\varepsilon}\) with respect to \(\theta\).
The population conditional flow therefore serves as a fixed-point guide for the projected training rule.
If the generator family cannot realize all conditional laws \(p_x\), the practical fixed point is the best representable conditional surrogate reached by this projected training dynamics.
For anchors \(x_i\sim\mu\), the estimator draws multiple generated samples
\(\hat y_{i,k}=g_\theta(x_i,z_{i,k})\), multiple positive channel samples
\(y_{i,j}\sim p(\cdot\mid x_i)\), and an independent generated reference batch
\(\hat y'_{i,\ell}=g_\theta(x_i,z'_{i,\ell})\).
Sinkhorn barycentric projections are then computed separately for each anchor:
\begin{equation}
v_{i,k}
=
T^{\varepsilon}_{\hat Q_i,P_i}(\hat y_{i,k})
-
T^{\varepsilon}_{\hat Q_i,\hat Q'_i}(\hat y_{i,k}),
\end{equation}
followed by the detached regression target
\begin{equation}
\tilde y_{i,k}
=
\mathrm{stopgrad}
\left(
\hat y_{i,k}+\eta v_{i,k}
\right).
\end{equation}
Table~\ref{tab:condition-wise-sinkhorn-algorithm} summarizes the resulting training update.

\begin{table}[t]
\caption{\textbf{Condition-wise Sinkhorn drifting update.}}
\label{tab:condition-wise-sinkhorn-algorithm}
\centering
\tablestyle{3.8pt}{1.04}
\footnotesize
\begin{tabular}{@{}p{0.96\columnwidth}@{}}
\toprule
\textbf{Input:} generator \(g_\theta\), channel sampler \(p(\cdot\mid x)\), anchors \(\{x_i\}_{i=1}^{B}\), generated count \(K_g\), positive count \(K_p\), reference count \(K_r\), drift scale \(\eta\). \\
\midrule
1. Draw \(\hat y_{i,k}=g_\theta(x_i,z_{i,k})\), \(k=1,\ldots,K_g\). \\
2. Draw positive samples \(y_{i,j}\sim p(\cdot\mid x_i)\), \(j=1,\ldots,K_p\). \\
3. Draw generated reference samples \(\hat y'_{i,\ell}=g_\theta(x_i,z'_{i,\ell})\), \(\ell=1,\ldots,K_r\). \\
4. For each anchor \(i\), solve two output-space Sinkhorn problems, \(\hat Q_i\!\leftrightarrow\!P_i\) and \(\hat Q_i\!\leftrightarrow\!\hat Q'_i\). \\
5. Form barycentric velocity \(v_{i,k}=T^\varepsilon_{\hat Q_i,P_i}(\hat y_{i,k})-T^\varepsilon_{\hat Q_i,\hat Q'_i}(\hat y_{i,k})\). \\
6. Regress \(g_\theta(x_i,z_{i,k})\) toward \(\mathrm{stopgrad}(\hat y_{i,k}+\eta v_{i,k})\). \\
\bottomrule
\end{tabular}
\end{table}

\paragraph{Local conditional approximation.}
The exact estimator above is natural for simulator channels where repeated draws at the same \(x_i\) are available.
Measured datasets often contain only one or a few observations per condition.
In that setting, the same geometry can be approximated by replacing the pointwise conditional law with a local kernel estimate in condition space.
Given measured pairs \(\{(x_r,y_r)\}_{r=1}^{N}\), define
\begin{equation}
\alpha_r(x_0)
=
\frac{K_h(x_r,x_0)}
{\sum_{\ell=1}^{N} K_h(x_\ell,x_0)},
\qquad
\hat p_{h,x_0}
=
\sum_{r=1}^{N}\alpha_r(x_0)\delta_{y_r},
\end{equation}
where \(K_h\) is a bandwidth-controlled kernel on transmitted symbols or channel-state features.
For a generated batch at anchor \(x_0\), condition-wise Sinkhorn is then replaced by a weighted local Sinkhorn problem between \(\{\hat y_k=g_\theta(x_0,z_k)\}\) and \(\hat p_{h,x_0}\).
Equivalently, a generated local reference law
\begin{equation}
\hat q_{\theta,h,x_0}
=
\sum_{r=1}^{N}\alpha_r(x_0)
\frac{1}{K}\sum_{k=1}^{K}\delta_{g_\theta(x_r,z_{r,k})}
\end{equation}
can be used when the learned generator should match the same local conditioning bandwidth as the measured target estimate.
The bandwidth \(h\) controls the usual bias--variance tradeoff: small neighborhoods preserve conditioning but increase estimator variance, while large neighborhoods reduce variance but approach a global transport problem.
The experiments in this paper use simulator access and therefore the exact repeated-condition estimator.
The local formulation is included only as a measured-data direction.
Bandwidth selection and neighborhood bias remain open.

The major operational advantage is straightforward.
Once training is complete, sampling requires only a single evaluation of \(g_\theta(x,z)\), which produces \(\hat y\) directly.
Sampling avoids reverse chains, ODE solvers, and per-sample iterative schedules at inference.

\section{Experimental Setup}
\subsection{Channel models}
We use four stochastic channel families.
All channels are implemented as real-valued maps \(x\mapsto y\), but the coordinate interpretation differs by channel.
AWGN and Rayleigh use real-valued vectors directly and have dimension \(n=7\) in the generator-level benchmark.
SSPA and TDL use paired in-phase/quadrature (I/Q) coordinates with \(n=8\), corresponding to four complex symbols.
Noise variables are sampled independently across real coordinates unless the deterministic channel transformation or tapped-delay convolution introduces coupling.
The generator-level comparison is reported on AWGN, Rayleigh, SSPA, and TDL.
AWGN, Rayleigh, and SSPA include the plain direct-drifting reference row, while the TDL drifting-family comparison focuses on the two Sinkhorn drift fields used in the final ablation.
For TDL, the DDPM, DDIM, and WGAN direct-output SWD rows are evaluated from the trained baseline checkpoints used in the SER-curve study.

\textbf{AWGN:}
\begin{equation}
y = x + n,
\qquad
n \sim \mathcal{N}(0,\sigma^2 I).
\end{equation}
This is the standard additive white Gaussian noise reference channel and serves as the simplest near-identity baseline.

\textbf{Rayleigh:}
the transmitted symbol is multiplied elementwise by a random fading amplitude and then corrupted by additive Gaussian noise:
\begin{equation}
y = h \odot x + n,
\end{equation}
where \(\odot\) denotes elementwise multiplication,
\begin{equation}
h_k = \frac{1}{\sqrt{2}}\sqrt{u_k^2 + v_k^2},
\qquad
u_k,v_k \overset{\text{i.i.d.}}{\sim} \mathcal{N}(0,1),
\end{equation}
and \(n \sim \mathcal{N}(0,\sigma^2 I)\).

Each real component is therefore scaled by an independent Rayleigh-distributed magnitude before noise is added.
The downstream encoder and decoder are not given \(h_k\), and no explicit equalization is applied.
As implemented here, Rayleigh is a real-valued vector channel.
The paired-I/Q complex representation is used for SSPA and TDL.

\textbf{SSPA:}
the input is passed through a smooth solid-state power amplifier nonlinearity before additive noise is applied.
In this case, the real-valued representation is interpreted as paired I/Q components.
Let
\begin{equation}
r = \sqrt{x_I^2 + x_Q^2}
\end{equation}
denote the input amplitude.
The multiplicative gain factor used in the benchmark is
\begin{equation}
a_{\mathrm{SSPA}}(r)
=
\frac{G}{\left(1+\left(\frac{Gr}{A_0}\right)^{2p}\right)^{1/(2p)}},
\end{equation}
with default parameters \(p=3\), \(A_0=1.5\), and \(G=5.0\).
Here \(a_{\mathrm{SSPA}}(r)\) is a gain, not the output amplitude.
The corresponding output amplitude is \(A_{\mathrm{out}}(r)=a_{\mathrm{SSPA}}(r)r\), the Rapp-style amplitude-to-amplitude law used in the benchmark.
The noiseless amplifier output is $\tilde{x} = a_{\mathrm{SSPA}}(r)\,x,$
and the observed channel output is
\begin{equation}
y = \tilde{x} + n,
\qquad
n \sim \mathcal{N}\!\left(0,\frac{\sigma^2}{2}I\right).
\end{equation}
This channel introduces amplitude-dependent nonlinear distortion while preserving the input direction in the I/Q plane.

\textbf{TDL:}
The TDL channel is a compact short-block approximation of the Third Generation Partnership Project (3GPP) Technical Report 38.901 TDL-D profile \cite{tr38901}.
The real vector is interpreted as a complex codeword \(s_m=x_{2m}+j x_{2m+1}\), \(m=0,\ldots,N_c-1\), with \(N_c=4\) complex channel uses in the symbolic coding runs.
For each codeword, the channel samples a finite set of complex taps and applies circular convolution,
\begin{equation}
\tilde{s}_m
=
\sum_{\ell=0}^{L_h-1}
h_\ell s_{(m-d_\ell)\bmod N_c}.
\end{equation}
The observed output is
\begin{equation}
y_m = \tilde{s}_m+n_m,
\qquad
n_m\sim \mathcal{CN}(0,\sigma^2).
\end{equation}
The implemented profile uses the TDL-D line-of-sight structure with first-tap \(K\)-factor \(13.3\,\mathrm{dB}\).
Because the codeword is short, the normalized path delays from the standard profile are rounded to a small set of effective circular shifts and path powers are renormalized after aggregation.
This keeps the channel compatible with the short symbolic autoencoder while preserving the random multipath mixture structure.
We use this TDL configuration as a compact channel-with-memory stress test.
It is intentionally smaller than a full OFDM or long-block TDL evaluation.

\subsection{Training protocol}
The experiments are organized into two complementary protocols.
First, the generator-level protocol compares one-shot drifting-family generators against established conditional DDPM, DDIM, and WGAN references on AWGN, Rayleigh, SSPA, and TDL.
Second, the drift-field ablation holds the drifting network fixed and compares global/joint Sinkhorn drifting with condition-wise Sinkhorn drifting.
The final tables keep the drift-field comparison focused on direct drifting and the two Sinkhorn transport fields.
This second protocol is reported over 100 seeds for the main drift-field runs and is evaluated with direct SWD, anchor-conditioned metrics, and downstream symbolic SER/BER.
It is the controlled comparison in this paper because the generator architecture, optimizer family, and one-shot inference path are held fixed across drift fields.
The DDPM, DDIM, and WGAN rows provide best-effort reference baselines using stable model-family settings.
A matched-capacity ranking against drifting would require an additional equal-parameter and equal-budget study.
For SSPA, we additionally report a 30-seed update-budget-controlled condition-wise Sinkhorn operating point because the anchor-conditioned diagnostics stabilize at a smaller update budget than the long diffusion-style SSPA preset.
For the direct-output SWD table, TDL W-Flow rows come from the drift-field ablation, and TDL DDPM/DDIM/WGAN rows are obtained by checkpoint-only direct-SWD evaluation of the corresponding trained baseline implants.

\begin{table*}[t]
\centering
\caption{\textbf{Simulation setup used for the journal experiments.} The listed \(E_b/N_0\) values are the nominal channel-model and autoencoder training points. Generator-level AWGN/Rayleigh/SSPA reference rows use ten seeds; TDL reference rows use 30 checkpoint evaluations. W-Flow rows use 100 seeds except for the 30-seed SSPA update-budget-controlled operating point. All SWD evaluations use 128 random projections.}
\label{tab:simulation-setup}
\tablestyle{4.2pt}{1.04}
\tablefontsize
\begin{tabular}{@{}lccccccc@{}}
\toprule
Channel & \(n\) & \(M_{\mathrm{msg}}\) & \(R\) & \(E_b/N_0\) & Train samples / batch / epochs & Eval. samples & Main use \\
\midrule
AWGN & 7 & 16 & \(4/7\) & \(5\,\mathrm{dB}\) & \(10^7 / 5000 / 30\) & \(10^6\) & Reference + W-Flow \\
Rayleigh & 7 & 16 & \(4/7\) & \(12\,\mathrm{dB}\) & \(10^7 / 5000 / 30\) & \(10^6\) & Reference + W-Flow \\
SSPA reference & 8 & 64 & \(6/8\) & \(8\,\mathrm{dB}\) & \(10^7 / 4096 / 160\) & \(10^6\) & DDPM/DDIM/WGAN + direct drifting \\
SSPA W-Flow operating point & 8 & 64 & \(6/8\) & \(8\,\mathrm{dB}\) & \(1.2{\times}10^5 / 4096 / 160\) & \(10^6\) & Reported condition-wise Sinkhorn \\
TDL & 8 & 16 & \(4/8\) & \(10\,\mathrm{dB}\) & \(1.2{\times}10^5 / 512 / 60\) & \(10^5\) & Reference + W-Flow \\
\bottomrule
\end{tabular}
\end{table*}

\paragraph{Channel and data settings.}
Table~\ref{tab:simulation-setup} collects the channel dimensions, message alphabets, training points, and generator-evaluation budgets.
The noise standard deviations induced by these \((E_b/N_0,R)\) pairs are \(0.5260\), \(0.2350\), \(0.3251\), and \(0.3162\) for AWGN, Rayleigh, SSPA, and TDL.
The SSPA condition-wise Sinkhorn row uses the compact update-budget-controlled operating point.
Table~\ref{tab:sspa-budget-sensitivity} reports the corresponding SSPA budget sensitivity across a small set of diagnostic checkpoints.
Downstream symbolic autoencoders are trained through either the analytic channel or a learned channel implant and are always evaluated on the analytic channel.
Within each channel, the autoencoder architecture, optimizer schedule, message alphabet, rate, training \(E_b/N_0\), and epoch budget are fixed; only the channel implant used during training changes.
Reported downstream uncertainty is taken over the complete seeded training and evaluation pipeline.

\paragraph{Model settings.}
Drifting uses a conditional MLP with latent dimension \(16\), hidden width \(128\), two hidden layers, \texttt{SiLU} activations, and Adam with learning rate \(10^{-3}\).
The base drifting objective uses drift scale \(1.0\), median-heuristic bandwidth selection, minimum bandwidth \(0.2\), maximum drift norm \(2.0\), and repulsive weight \(1.0\).
The Sinkhorn drift-field variants use the same direct-output generator architecture and detached-target update, but replace the training drift field.
The joint Sinkhorn row uses a Sinkhorn/W-Flow drift in joint condition-output features, while the condition-wise row uses the per-anchor conditional Sinkhorn field in \eqref{eq:condition_sinkhorn_velocity}.
Unless otherwise stated, Sinkhorn variants use \(10\) Sinkhorn iterations, automatic entropic scale selection with minimum \(\varepsilon=10^{-3}\), four generated samples, four positive channel samples, and four generated reference samples per condition anchor.

For AWGN, Rayleigh, SSPA, and TDL, diffusion models \(p(y\mid x)\) directly with \(v\)-prediction, an unclipped cosine variance schedule, \(100\) diffusion steps, and exponential moving average decay \(0.9\). Hidden widths are \(110\) (AWGN and SSPA) and \(128\) (Rayleigh and TDL). AWGN and Rayleigh use a two-stage learning-rate schedule (\(10\) epochs at \(10^{-3}\), then \(20\) epochs at \(10^{-4}\)), while SSPA and TDL use a constant \(10^{-4}\) for \(160\) and \(60\) epochs, respectively. We report DDPM and DDIM with \(100\), \(50\), \(20\), and \(10\) denoising steps.

WGAN uses a conditional Wasserstein objective with RMSprop at \(10^{-4}\) for generator and critic, five critic updates per generator update, and weight clipping at \(0.01\). Hidden widths are \(128\) (AWGN) and \(256\) (Rayleigh, SSPA, and TDL). WGAN is reported in direct output space wherever it is compared against direct-output channel generators.

\paragraph{Experimental fairness.}
The benchmark is intended as a practical comparative study, with stable model-family implementations prioritized over strict matched-capacity control.
Architectures and training schedules were chosen as stable implementations for each model family.
Identical parameter counts and compute budgets across drifting, diffusion, and WGAN are outside the benchmark design.
The Sinkhorn rows are therefore reported as a drift-field ablation under a matched drifting architecture and training budget.
They are compared to diffusion and WGAN wherever the observable metric and channel protocol overlap.
The downstream W-Flow table focuses on the learned drifting/Sinkhorn surrogates and the analytic-channel reference training condition.
The SER curves add WGAN and DDIM-100 when all corresponding channel implants are available.
The central ablation is the drift-field geometry: joint/global transport versus condition-wise Sinkhorn transport under the same one-shot generator.

As an additional diagnostic, Table~\ref{tab:model-size} reports inference-network parameter counts used in this benchmark.
The Sinkhorn rows have the same parameter count as direct drifting because the transport computation changes the training update, not the deployed generator.

\begin{table}[t]
\centering
\caption{\textbf{Model parameter counts in the benchmark setup.} Counts refer to inference-time networks. Sinkhorn drifting uses the same one-shot generator as direct drifting. Its Sinkhorn couplings are training-time computations and add no inference-time parameters. DDPM and DDIM share the same denoiser weights. WGAN rows report generator-side parameters.}
\label{tab:model-size}
\tablestyle{4pt}{1.03}
\tablefontsize
\begin{tabular}{@{}lcccc@{}}
\toprule
Model & AWGN & Rayleigh & SSPA & TDL \\
\midrule
Direct drifting & 20{,}487 & 20{,}487 & 20{,}744 & 20{,}744 \\
Sinkhorn drifting variants & 20{,}487 & 20{,}487 & 20{,}744 & 20{,}744 \\
Diffusion denoiser (DDPM / DDIM) & 59{,}847 & 74{,}247 & 60{,}178 & 74{,}632 \\
WGAN generator & 19{,}335 & 71{,}431 & 72{,}200 & 72{,}200 \\
\bottomrule
\end{tabular}
\end{table}

\paragraph{Seeding and reporting.}
  The diffusion/WGAN reference benchmark results are reported as mean \(\pm\) standard deviation over ten seeds.
  The Sinkhorn drift-field ablation is reported over 100 seeds unless stated otherwise. The SSPA update-budget check uses 30 seeds.
  We use standard deviation for the ten-seed reference rows to describe run-to-run variability.
  W-Flow tables use mean \(\pm\) standard error because the main purpose of those rows is to compare seed-averaged drift-field variants.
  Each run uses deterministic initialization to reduce run-to-run variance unrelated to model stochasticity.

\subsection{Evaluation metric}
The diffusion/WGAN/drifting reference comparison emphasizes generator-level distributional accuracy and timing, while the Sinkhorn drift-field ablation adds anchor-conditioned diagnostics and downstream symbolic SER/BER.
The primary generator metric is direct-output sliced Wasserstein distance (SWD).
For generated channel outputs, we report \(\mathrm{SWD}\!\left(\{\hat{y}_i\}, \{y_i\}\right)\).
Given true samples \(P=\{v_i\}_{i=1}^N\) and generated samples \(Q=\{\hat v_i\}_{i=1}^N\), SWD is estimated by drawing random unit directions \(\theta_m\), projecting both clouds onto those directions, computing the one-dimensional Wasserstein distance for each projection, and averaging over projections:
\begin{equation*}
\mathrm{SWD}(P,Q)
=
\frac{1}{M}
\sum_{m=1}^M
W_1\!\left(\{\langle \theta_m,v_i\rangle\}_{i=1}^N,\{\langle \theta_m,\hat v_i\rangle\}_{i=1}^N\right).
\end{equation*}
Thus SWD compares the generated and true sample clouds through many one-dimensional views and captures geometric mismatch more faithfully than a purely bin-based marginal comparison. In our experiments, SWD is estimated using \(128\) random one-dimensional projections on AWGN,
  Rayleigh, SSPA, and TDL.
The projection directions are sampled from an isotropic Gaussian distribution, normalized to unit length, and generated from the deterministic metric seed associated with the run.
SWD is used primarily for comparability with prior diffusion-channel benchmarks \cite{paper2309,kimicc2024}. Compared with histogram-based \(L_1\) distance, it also better captures sample geometry in \(\mathbb{R}^n\).
SWD provides a geometry-aware comparison of global output sample clouds.
It leaves decision-relevant behavior such as error rates unmeasured.
Those depend on the joint structure of \((x,y)\).
For communication use, two generators with similar global sample-cloud discrepancy can induce different conditional means, conditional covariance errors, and downstream symbol-error rates.
The Sinkhorn drift-field study therefore also includes anchor-conditioned moment metrics and symbolic channel-implant SER/BER checks when evaluating the proposed drifting variants.
We therefore report both downstream-dependent and generator-level metrics.
Downstream SER and BER under analytic-channel evaluation measure the utility of a surrogate inside the trained communication system, but they also depend on the chosen encoder, decoder, optimizer, and training point.
SWD is independent of a downstream receiver and remains useful for comparing output distributions across generators, while global SWD can miss fixed-input conditional errors.
Anchor-conditioned metrics bridge these views by probing repeated outputs at fixed transmitted symbols.

\section{Benchmark Results}
\subsection{Unified SWD benchmark across channels}
Table~\ref{tab:wflow-swd-baselines} reports a unified direct-output SWD comparison across AWGN, Rayleigh, SSPA, and TDL.
The reference rows summarize the diffusion and WGAN baselines.
The drifting-family rows contain direct drifting and the two Sinkhorn drift fields under a matched one-shot architecture.
The controlled conclusion from this table is the ranking inside the drifting-family block, while the diffusion and WGAN rows provide reference context from their own stable training protocols.
DDIM-100 gives the lowest direct-space SWD on AWGN, Rayleigh, and SSPA among the matched reference rows.
Among the Sinkhorn variants, condition-wise Sinkhorn is strongest under direct SWD on AWGN, Rayleigh, and TDL, and is numerically close to direct drifting on SSPA under the compact update-budget run.
Within the full drifting-family block, condition-wise Sinkhorn is lowest on AWGN, Rayleigh, and TDL, while direct drifting remains lowest on SSPA under this global metric.
The drift-field ranking is therefore channel-dependent.
\begin{table*}[t]
\centering
\caption{\textbf{Direct-output SWD comparison for reference baselines and main drifting variants.} Lower is better. Reference rows report mean $\pm$ standard deviation from the diffusion/WGAN benchmark or from checkpoint-only direct-SWD evaluation. The drifting-family rows retain direct drifting and the two Sinkhorn transport fields used for the main ablation. W-Flow rows report the available-seed mean $\pm$ standard error. Bold marks the best drifting-family row per channel under the reported mean. Dashes mark unavailable matched generator-level SWD values.}
\label{tab:wflow-swd-baselines}
\tablestyle{5.8pt}{1.08}
\footnotesize
\begin{tabular}{@{}lcccc@{}}
\toprule
Method & AWGN & Rayleigh & SSPA & TDL \\
\midrule
\rowcolor[gray]{0.9} \multicolumn{5}{l}{\textit{Diffusion/WGAN reference baselines}} \\
\headspace WGAN & $0.0198\err{0.0050}$ & $0.0171\err{0.0050}$ & $0.0287\err{0.0156}$ & $0.0289\err{0.0063}$ \\
\headspace DDPM & $0.0070\err{0.0004}$ & $0.0079\err{0.0008}$ & $0.0032\err{0.0005}$ & $0.0297\err{0.0011}$ \\
\headspace DDIM-100 & $0.0042\err{0.0006}$ & $0.0044\err{0.0007}$ & $0.0024\err{0.0004}$ & $0.0248\err{0.0009}$ \\
\midrule
\rowcolor[gray]{0.9} \multicolumn{5}{l}{\textit{Drifting family}} \\
\headspace Direct drifting & $0.0100\err{0.0007}$ & $0.0085\err{0.0012}$ & $\mathbf{0.0060}\err{0.0008}$ & -- \\
\headspace Joint Sinkhorn & $0.0148\err{0.0002}$ & $0.0152\err{0.0002}$ & $0.0117\err{0.0002}$ & $0.0966\err{0.0005}$ \\
\headspace Condition-wise Sinkhorn & $\mathbf{0.0058}\err{0.0001}$ & $\mathbf{0.0073}\err{0.0001}$ & $0.0071\err{0.0002}$ & $\mathbf{0.0118}\err{0.0002}$ \\
\bottomrule
\end{tabular}
\end{table*}

Timing follows the expected pattern.
Drifting and WGAN operate at roughly \(0.1\,\mu\mathrm{s}\) per sample on the local graphics processing unit (GPU) timing setup, while DDPM and DDIM require microsecond-to-tens-of-microseconds sampling times depending on the sampler and the number of denoising steps.
Thus, even when diffusion achieves the strongest SWD, one-shot generators retain a substantial practical latency advantage.
The W-Flow variants retain the same one-shot inference path as direct drifting.
Their cost differences are training-time drift-field costs.
Table~\ref{tab:timing-cuda} therefore reports W-Flow ablation rows under the same \(2\%\) local GPU timing protocol.
The timing run uses PyTorch on a local NVIDIA RTX 5060 Ti GPU in standard float32 execution, without CUDA graph capture, \texttt{torch.compile}, or mixed precision.
Training time is measured by running the actual training loop on \(2\%\) of the configured dataset size, rounded to full batches, and extrapolating by the ratio of timed optimizer steps to full optimizer steps.
This includes synthetic channel sampling, forward and backward passes, optimizer updates, and the Sinkhorn computations used by W-Flow rows.
It excludes checkpoint writing and downstream symbolic autoencoder training.
Inference timing uses batch size \(512\), \(40\) batches per repeat, two warmup repeats, seven timed repeats, and CUDA synchronization around each repeat.
It includes drawing the conditioning inputs and evaluating the generator.
For diffusion rows, inference includes the complete DDPM or DDIM denoising trajectory and the same conditioning overhead.
The total column adds projected training time to the time needed to generate the evaluation-sample budget listed in Table~\ref{tab:simulation-setup}.
On AWGN, Rayleigh, and SSPA, the condition-wise Sinkhorn row is faster to train than joint Sinkhorn because it solves many small same-condition transport problems instead of forming a global minibatch interaction matrix.
This agrees with the leading coupling cost: with \(B\) anchors, \(K_g\) generated particles, \(K_p\) positives, \(K_r\) generated references, and \(I\) Sinkhorn iterations, condition-wise Sinkhorn scales as \(O(I B K_g(K_p+K_r))\), while joint Sinkhorn over the expanded batch scales as \(O(I B^2 K_g(K_p+K_r))\).
At inference, W-Flow, direct drifting, and WGAN retain one generator call per sample, whereas DDPM/DDIM use \(T\) denoising-network evaluations.

\begin{table*}[t]
\centering
\caption{\textbf{Projected training and inference timing on an NVIDIA RTX 5060 Ti GPU.} Training hours are extrapolated from a 2\% timing run. Inference is reported as time per sample.}
\label{tab:timing-cuda}
\tablestyle{3pt}{1.02}
\tablefontsize
\resizebox{\textwidth}{!}{%
\begin{tabular}{@{}lcccccccccccc@{}}
\toprule
Method & \multicolumn{3}{c}{AWGN} & \multicolumn{3}{c}{Rayleigh} & \multicolumn{3}{c}{SSPA} & \multicolumn{3}{c}{TDL} \\
\cmidrule(lr){2-4} \cmidrule(lr){5-7} \cmidrule(lr){8-10} \cmidrule(lr){11-13}
 & Train [h] & Inf. & Total [h] & Train [h] & Inf. & Total [h] & Train [h] & Inf. & Total [h] & Train [h] & Inf. & Total [h] \\
\midrule
\rowcolor[gray]{0.9} \multicolumn{13}{l}{\textit{One-shot generators}} \\
Drifting (dir.) & 1.299 & 0.114 us & 1.300 & 1.299 & 0.117 us & 1.300 & 5.656 & 0.118 us & 5.656 & 0.007 & 0.116 us & \textbf{0.007} \\
Joint Sinkhorn & 0.917 & 0.116 us & 0.918 & 0.919 & 0.115 us & 0.920 & 3.956 & 0.117 us & 3.956 & 0.010 & 0.116 us & 0.010 \\
Condition-wise Sinkhorn & 0.048 & 0.115 us & \textbf{0.048} & 0.048 & 0.115 us & \textbf{0.049} & 0.296 & 0.117 us & \textbf{0.297} & 0.009 & 0.115 us & 0.009 \\
WGAN & 0.064 & 0.108 us & 0.064 & 0.085 & 0.108 us & 0.086 & 0.478 & 0.107 us & 0.478 & 0.022 & 0.108 us & 0.022 \\
\midrule
\rowcolor[gray]{0.9} \multicolumn{13}{l}{\textit{Diffusion samplers}} \\
DDPM & 0.020 & 0.046 ms & 0.147 & 0.020 & 0.045 ms & 0.145 & 0.124 & 0.047 ms & 0.253 & 0.005 & 0.045 ms & 0.007 \\
DDIM-100 & 0.020 & 0.039 ms & 0.128 & 0.020 & 0.039 ms & 0.128 & 0.124 & 0.039 ms & 0.231 & 0.005 & 0.039 ms & 0.006 \\
DDIM-50 & 0.020 & 0.019 ms & 0.073 & 0.020 & 0.019 ms & 0.073 & 0.124 & 0.019 ms & 0.177 & 0.005 & 0.019 ms & 0.006 \\
DDIM-20 & 0.020 & 0.008 ms & 0.042 & 0.020 & 0.008 ms & 0.042 & 0.124 & 0.008 ms & 0.145 & 0.005 & 0.008 ms & 0.006 \\
DDIM-10 & 0.020 & 0.004 ms & 0.031 & 0.020 & 0.004 ms & 0.031 & 0.124 & 0.004 ms & 0.135 & 0.005 & 0.004 ms & 0.005 \\
\bottomrule
\end{tabular}
}
\par\vspace{0.25em}
\begin{minipage}{0.96\textwidth}
\footnotesize\raggedright
All rows use the same 2\% extrapolation protocol. In the one-shot block, bold total times mark the lowest projected total for each channel. The W-Flow rows use the same one-shot generator architecture as direct drifting. Their training-time differences come only from the drift-field computation.
\end{minipage}
\end{table*}

\subsection{Condition-wise drift-field interpretation}
The W-Flow-inspired ablation isolates the drift-field design problem.
Joint Sinkhorn drift and condition-wise Sinkhorn drift use the same one-shot generator architecture and differ only in how training particles are moved.
Geometrically, joint Sinkhorn builds a transport field from \((x,y)\)-features, while condition-wise Sinkhorn solves separate output-space transport problems at fixed \(x\).
This distinction is visible in the mixed metric behavior.
Condition-wise Sinkhorn is the strongest W-Flow row under direct SWD on AWGN, Rayleigh, and TDL, while SSPA needs the compact update-budget run because the corrected same-condition field converges earlier than the long diffusion-style preset.
Fig.~\ref{fig:conditional-fiber-diagnostic} shows the corresponding fixed-condition diagnostic for SSPA.
For three transmitted-symbol anchors, the joint Sinkhorn field misses both the conditional mean and the local spread of the analytic output fiber.
Condition-wise Sinkhorn keeps the generated cloud close to the analytic same-\(x\) cloud.
In this channel-simulator setting, the downstream checks below are the more relevant test.
Accordingly, global SWD is treated as a comparability diagnostic and supplemented by conditional or downstream metrics for model selection.

\begin{figure*}[t]
\centering
\includegraphics[width=0.66\textwidth]{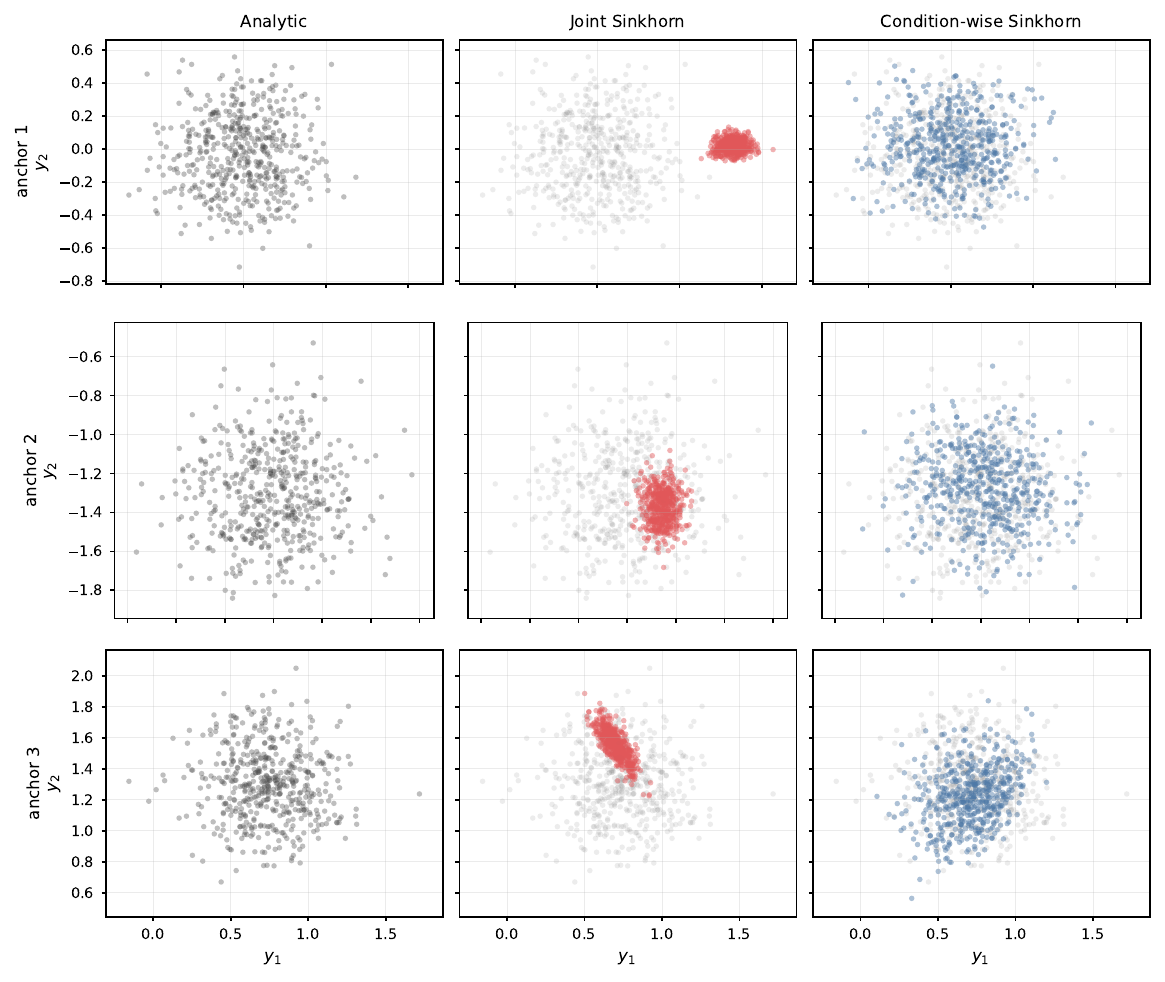}
\caption{\textbf{Fixed-condition SSPA output fibers.} Each row fixes one transmitted-symbol anchor \(x_i\) and plots \(512\) repeated channel outputs in the first I/Q plane; the three anchors are selected from \(256\) candidate transmitted symbols. The analytic channel is shown in gray. The learned columns overlay generated samples from joint Sinkhorn and condition-wise Sinkhorn on the same analytic cloud. Joint transport can match aspects of the global output distribution while displacing the fixed-\(x\) fiber. Condition-wise Sinkhorn better preserves the conditional output law.}
\label{fig:conditional-fiber-diagnostic}
\end{figure*}

\subsection{Conditional metrics and symbolic coding}
Table~\ref{tab:wflow-conditional-metrics} reports the corresponding anchor-conditioned diagnostics.
The anchor metrics draw repeated outputs at fixed transmitted symbols and therefore probe the conditional law more directly than global SWD.
Direct SWD is reported in Table~\ref{tab:wflow-swd-baselines}, and downstream BER/SER is reported in Table~\ref{tab:wflow-ser-ber}.
For AWGN, Rayleigh, SSPA, and TDL, condition-wise Sinkhorn is the strongest evaluated drifting-family surrogate under the downstream BER/SER checks reported here.
The full curve comparison below shows that diffusion remains the strongest learned reference on Rayleigh, SSPA, and TDL.
On SSPA, the anchor-conditioned metrics motivate the compact update-budget run.
The full-preset row in Table~\ref{tab:sspa-budget-sensitivity} is retained as a sensitivity check for this saturated nonlinear channel.

\begin{table}[!t]
\centering
\caption{\textbf{Anchor-conditioned diagnostics for Sinkhorn drift fields.} Anchor SWD and Gaussian Wasserstein-2 (GW2) are computed from repeated samples at fixed channel inputs. Direct SWD and BER/SER are reported separately in Tables~\ref{tab:wflow-swd-baselines} and~\ref{tab:wflow-ser-ber}. Values are seed means. Lower is better.}
\label{tab:wflow-conditional-metrics}
\tablestyle{3.4pt}{1.02}
\tablefontsize
\begin{tabular}{@{}llcc@{}}
\toprule
Channel & Variant & Anchor SWD & Anchor GW2 \\
\midrule
AWGN & Joint Sinkhorn & 0.1199 & 0.4443 \\
AWGN & Condition-wise Sinkhorn & 0.1154 & 0.4277 \\
\midrule
Rayleigh & Joint Sinkhorn & 0.1723 & 0.6929 \\
Rayleigh & Condition-wise Sinkhorn & 0.1097 & 0.3820 \\
\midrule
SSPA & Joint Sinkhorn & 0.2174 & 0.7876 \\
SSPA & Condition-wise Sinkhorn & 0.0566 & 0.2376 \\
\midrule
TDL & Joint Sinkhorn & 0.3221 & 1.1871 \\
TDL & Condition-wise Sinkhorn & 0.1054 & 0.4444 \\
\bottomrule
\end{tabular}
\par\vspace{0.25em}
\begin{minipage}{0.92\columnwidth}
\footnotesize\raggedright
The SSPA condition-wise Sinkhorn row uses the same \(30\)-seed update-budget-controlled run as Table~\ref{tab:wflow-swd-baselines}.
\end{minipage}
\end{table}

Table~\ref{tab:sspa-budget-sensitivity} makes the SSPA budget choice explicit.
The compact condition-wise Sinkhorn result is aggregated over the same 30 seeds used in Tables~\ref{tab:wflow-swd-baselines} and~\ref{tab:wflow-ser-ber}.
The first two rows are single-seed local sanity checks, and the full-preset row is a sensitivity diagnostic rather than a model-selection candidate.

\begin{table*}[t]
\centering
\caption{\textbf{SSPA condition-wise Sinkhorn budget sensitivity.} The table reports the corrected condition-wise Sinkhorn field under increasing optimizer budgets. Updates are approximate optimizer steps, computed from dataset size, batch size, and epochs. Ratios divide the learned-vs-analytic anchor metric by the analytic-vs-analytic floor. The floor is estimated from two independent analytic channel sample sets at the same anchors. The first two rows are single-seed local diagnostics; the last two rows are aggregated cluster runs. Lower is better.}
\label{tab:sspa-budget-sensitivity}
\tablestyle{4.0pt}{1.03}
\tablefontsize
\begin{tabular}{@{}lcccccc@{}}
\toprule
Run & Updates & Seeds & Direct SWD & Anchor-SWD ratio & Anchor-GW2 ratio & SER at \(8\,\mathrm{dB}\) \\
\midrule
Short local & \(0.31\mathrm{k}\) & 1 & \(1.07{\times}10^{-1}\) & 5.14 & 4.10 & -- \\
Early compact & \(0.88\mathrm{k}\) & 1 & \(4.09{\times}10^{-2}\) & 4.50 & 3.88 & -- \\
Reported operating point & \(4.69\mathrm{k}\) & 30 & \(7.07{\times}10^{-3}\err{1.9{\times}10^{-4}}\) & \(1.13\err{0.005}\) & \(1.15\err{0.004}\) & \(9.33{\times}10^{-5}\err{1.4{\times}10^{-5}}\) \\
Full preset & \(390.6\mathrm{k}\) & 100 & \(1.53\err{0.02}\) & \(35.5\err{0.4}\) & \(60.1\err{0.7}\) & \(2.70{\times}10^{-2}\err{2.2{\times}10^{-3}}\) \\
\bottomrule
\end{tabular}
\par\vspace{0.25em}
\begin{minipage}{0.94\textwidth}
\footnotesize\raggedright
The reported SSPA condition-wise Sinkhorn model is the \(4.69\mathrm{k}\)-update operating point selected by anchor-conditioned diagnostics. The full-preset row is included only as a sensitivity diagnostic.
\end{minipage}
\end{table*}

The downstream symbolic results in Table~\ref{tab:wflow-ser-ber} use the learned channel surrogate for autoencoder training and the analytic channel for final evaluation.
The analytic row reports autoencoder training through the analytic channel.
Because training and evaluation are finite seeded runs, it can still deviate slightly from an ideal lower-bound interpretation.
We treat SER as the primary symbolic metric because the autoencoder message alphabet is finite.
BER is retained in the table as an additional bit-level diagnostic.
Within the Sinkhorn ablation table, condition-wise Sinkhorn is the strongest downstream drifting variant when the SSPA row uses the update-budget-controlled operating point.
The per-channel SER curves in Fig.~\ref{fig:wflow-all-ser-curves} add WGAN and diffusion references to this downstream comparison.
At the nominal training point, condition-wise Sinkhorn nearly reaches the analytic AWGN reference, while diffusion is the strongest learned baseline on Rayleigh, SSPA, and TDL.
On SSPA, the \(M_{\mathrm{msg}}=64\) update-budget-controlled operating point gives condition-wise Sinkhorn SER \(8.13\times10^{-5}\) at \(E_b/N_0=8\,\mathrm{dB}\), compared with \(7.50\times10^{-5}\) for WGAN, \(3.13\times10^{-5}\) for diffusion, and \(1.57\times10^{-5}\) for analytic-channel training.
The long-block check below tests the same AWGN implant outside the short symbolic-codeword setting.

\begin{table*}[t]
\centering
\caption{\textbf{Downstream symbolic coding metrics for Sinkhorn channel surrogates.} Autoencoders are trained through each channel surrogate and evaluated on the analytic channel. Values are mean $\pm$ standard error over seeds. Lower is better. Bold marks the best learned surrogate in each row, while the analytic channel is the non-learned reference training condition.}
\label{tab:wflow-ser-ber}
\tablestyle{3.4pt}{1.02}
\tablefontsize
\begin{tabular}{@{}llccc@{}}
\toprule
Channel & Metric & Analytic & Joint Sinkhorn & Condition-wise Sinkhorn \\
\midrule
AWGN & BER & $0.00293\err{4.0\times 10^{-5}}$ & $0.00388\err{5.5\times 10^{-5}}$ & $\mathbf{0.00302}\err{4.2\times 10^{-5}}$ \\
AWGN & SER & $0.00551\err{7.3\times 10^{-5}}$ & $0.00729\err{9.7\times 10^{-5}}$ & $\mathbf{0.00568}\err{7.5\times 10^{-5}}$ \\
\midrule
Rayleigh & BER & $0.00273\err{4.3\times 10^{-5}}$ & $0.0177\err{3.3\times 10^{-4}}$ & $\mathbf{0.00471}\err{8.2\times 10^{-5}}$ \\
Rayleigh & SER & $0.00514\err{8.0\times 10^{-5}}$ & $0.0331\err{5.6\times 10^{-4}}$ & $\mathbf{0.00881}\err{1.4\times 10^{-4}}$ \\
\midrule
SSPA & BER & $6.72\times 10^{-6}\err{1.9\times 10^{-6}}$ & -- & $\mathbf{4.90\times 10^{-5}}\err{7.5\times 10^{-6}}$ \\
SSPA & SER & $1.30\times 10^{-5}\err{3.5\times 10^{-6}}$ & -- & $\mathbf{9.33\times 10^{-5}}\err{1.4\times 10^{-5}}$ \\
\midrule
TDL & BER & $0.00158\err{2.1\times 10^{-5}}$ & $0.0107\err{3.7\times 10^{-4}}$ & $\mathbf{0.00293}\err{4.4\times 10^{-5}}$ \\
TDL & SER & $0.00297\err{3.6\times 10^{-5}}$ & $0.02\err{6.3\times 10^{-4}}$ & $\mathbf{0.00546}\err{7.6\times 10^{-5}}$ \\
\bottomrule
\end{tabular}
\par\vspace{0.25em}
\begin{minipage}{0.92\textwidth}
\footnotesize\raggedright
SSPA uses the \(30\)-seed \(M_{\mathrm{msg}}=64\) update-budget-controlled operating point. Dashes mark variants not rerun under that SSPA coding setup.
\end{minipage}
\end{table*}

For SSPA, we follow the diffusion-channel setup of \cite{paper2309} more closely and use message alphabet \(M_{\mathrm{msg}}=64\), \(n=8\), rate \(6/8\), and training at \(E_b/N_0=8\,\mathrm{dB}\).
The 30-seed SSPA operating point in Table~\ref{tab:wflow-ser-ber} shows that the corrected condition-wise Sinkhorn implementation is reliable once the W-Flow update budget is matched to the observed convergence of the sharp same-condition field.
Fig.~\ref{fig:wflow-all-ser-curves} reports SER over an \(E_b/N_0\) grid for AWGN, Rayleigh, SSPA, and TDL, which also shows where diffusion retains an advantage over one-shot drifting in downstream coding performance.


\begin{figure*}[!t]
\centering
\includegraphics[width=0.82\textwidth]{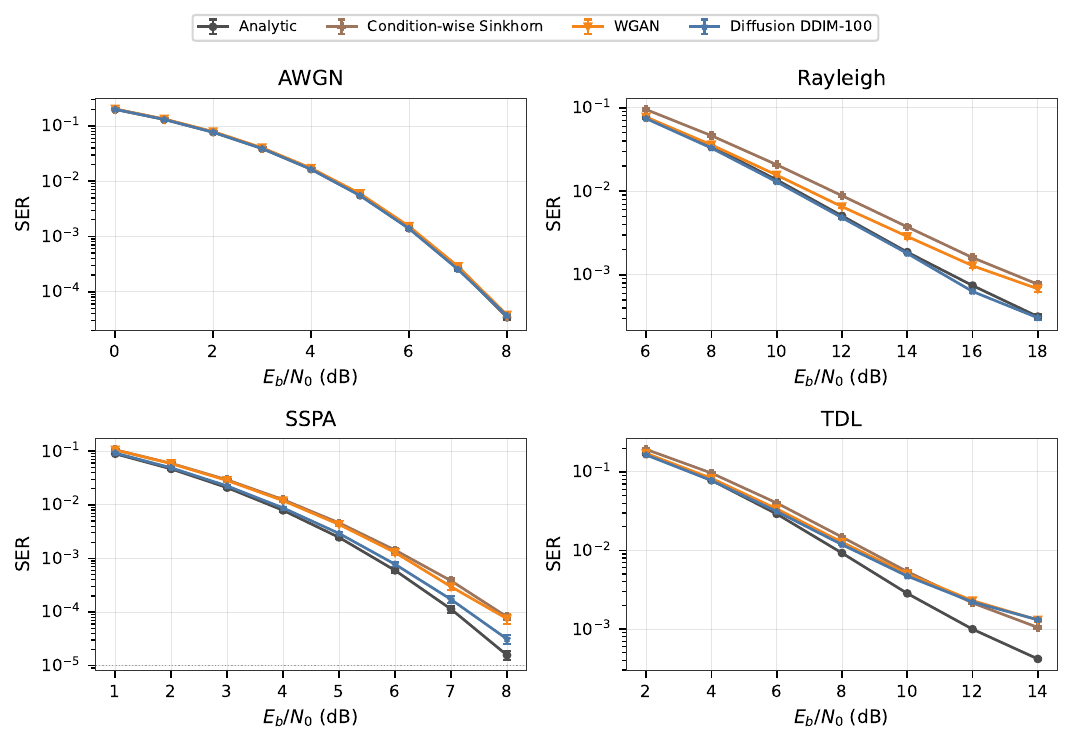}
\caption{\textbf{SER curves for learned channel surrogates.} Symbolic autoencoders are trained through each surrogate at the nominal point and evaluated on the analytic channel over an $E_b/N_0$ grid. Curves show seed means with standard-error bars; floor-clipped points are upper bounds.}
\label{fig:wflow-all-ser-curves}
\end{figure*}

\subsection{Equal wall-clock channel-implant training}
The previous downstream tables compare final trained channel implants.
The equal-wall-clock experiment fixes the symbolic autoencoder training time to \(1800\,\mathrm{s}\) per implant and records how many optimizer updates fit inside that budget.
This directly measures the inner-loop setting that motivates one-shot channel surrogates.
DDIM-100 is a strong learned reference in final-performance comparisons, but its iterative sampler leaves one to two orders of magnitude fewer autoencoder updates within the same budget.
Under this equal wall-clock budget, condition-wise Sinkhorn gives the lowest learned-implant SER on AWGN and SSPA, DDIM-100 is lowest on Rayleigh, and WGAN is lowest on compact TDL.
Condition-wise Sinkhorn reaches a competitive downstream operating point while retaining one-shot channel calls, although diffusion or WGAN can remain stronger on individual channels.

\begin{table*}[t]
\centering
\caption{\textbf{Equal wall-clock symbolic autoencoder training.} Each row trains the same symbolic autoencoder for \(1800\,\mathrm{s}\) through the listed channel implant and evaluates the resulting encoder/decoder on the analytic channel. Results are mean \(\pm\) standard error over 30 seeds.}
\label{tab:equal-wallclock-symbolic}
\tablestyle{3pt}{1.02}
\tablefontsize
\begin{tabular}{@{}llccc@{}}
\toprule
Channel & Training implant & Updates [k] & SER & BER \\
\midrule
AWGN & Analytic & 705.7 & $3.54\times 10^{-3}\err{3.94\times 10^{-5}}$ & $1.89\times 10^{-3}\err{2.31\times 10^{-5}}$ \\
& Condition-wise Sinkhorn & 626.8 & $\mathbf{3.62\times 10^{-3}}\err{3.65\times 10^{-5}}$ & $1.93\times 10^{-3}\err{2.31\times 10^{-5}}$ \\
& WGAN & 641.8 & $3.70\times 10^{-3}\err{4.06\times 10^{-5}}$ & $1.98\times 10^{-3}\err{2.87\times 10^{-5}}$ \\
& DDIM-100 & 21.5 & $5.32\times 10^{-3}\err{1.36\times 10^{-4}}$ & $2.83\times 10^{-3}\err{6.82\times 10^{-5}}$ \\
\midrule
Rayleigh & Analytic & 656.3 & $2.25\times 10^{-3}\err{4.57\times 10^{-5}}$ & $1.22\times 10^{-3}\err{2.94\times 10^{-5}}$ \\
& Condition-wise Sinkhorn & 602.1 & $5.40\times 10^{-3}\err{1.11\times 10^{-4}}$ & $2.92\times 10^{-3}\err{7.00\times 10^{-5}}$ \\
& WGAN & 614.0 & $4.83\times 10^{-3}\err{2.87\times 10^{-4}}$ & $2.55\times 10^{-3}\err{1.52\times 10^{-4}}$ \\
& DDIM-100 & 21.3 & $\mathbf{4.64\times 10^{-3}}\err{1.08\times 10^{-4}}$ & $2.47\times 10^{-3}\err{6.10\times 10^{-5}}$ \\
\midrule
SSPA & Analytic & 635.6 & $1.33\times 10^{-6}\err{6.31\times 10^{-7}}$ & $7.78\times 10^{-7}\err{3.81\times 10^{-7}}$ \\
& Condition-wise Sinkhorn & 617.1 & $\mathbf{8.80\times 10^{-5}}\err{1.87\times 10^{-5}}$ & $4.36\times 10^{-5}\err{8.30\times 10^{-6}}$ \\
& WGAN & 627.3 & $9.83\times 10^{-5}\err{3.15\times 10^{-5}}$ & $4.50\times 10^{-5}\err{1.27\times 10^{-5}}$ \\
& DDIM-100 & 21.4 & $9.00\times 10^{-5}\err{4.64\times 10^{-5}}$ & $4.03\times 10^{-5}\err{1.72\times 10^{-5}}$ \\
\midrule
TDL & Analytic & 449.0 & $1.30\times 10^{-3}\err{1.98\times 10^{-5}}$ & $6.98\times 10^{-4}\err{1.19\times 10^{-5}}$ \\
& Condition-wise Sinkhorn & 576.8 & $3.44\times 10^{-3}\err{6.05\times 10^{-5}}$ & $1.82\times 10^{-3}\err{3.43\times 10^{-5}}$ \\
& WGAN & 587.6 & $\mathbf{3.21\times 10^{-3}}\err{4.36\times 10^{-4}}$ & $1.70\times 10^{-3}\err{2.19\times 10^{-4}}$ \\
& DDIM-100 & 21.3 & $4.64\times 10^{-3}\err{1.23\times 10^{-4}}$ & $2.47\times 10^{-3}\err{6.36\times 10^{-5}}$ \\
\bottomrule
\end{tabular}
\par\vspace{0.25em}
\begin{minipage}{0.96\textwidth}
\footnotesize\raggedright
Bold marks the best learned implant per channel; the analytic row is the non-learned reference. The DDIM row uses 100 denoising steps per channel call, which yields many fewer optimizer updates within the same autoencoder-training time budget.
\end{minipage}
\end{table*}

\subsection{Long-block Turbo autoencoder (TurboAE) AWGN check}
We additionally test a convolutional neural network (CNN) TurboAE-style encoder--decoder at block length \(64\), rate \(1/2\), and AWGN training point \(E_b/N_0=4\,\mathrm{dB}\).
The TurboAE models are trained through either analytic AWGN or the condition-wise Sinkhorn AWGN surrogate and then evaluated on analytic AWGN.
This complements the short-block multi-channel SER/BER study with a longer differentiable coding loop.

Fig.~\ref{fig:turboae-long-block-curve} shows the resulting BER and block error rate (BLER) curves over \(30\) seeds.
At the \(4\,\mathrm{dB}\) training point, analytic-channel training gives BER \(1.96\times 10^{-3}\) and BLER \(8.18\times 10^{-2}\), with 95\% confidence intervals (CIs) \(5.21\times 10^{-4}\) and \(2.30\times 10^{-2}\).
Training through the condition-wise Sinkhorn surrogate gives BER \(3.30\times 10^{-3}\) and BLER \(1.39\times 10^{-1}\), with 95\% CIs \(4.30\times 10^{-4}\) and \(1.91\times 10^{-2}\).
The surrogate preserves the waterfall trend, with a measurable gap to analytic-channel training.

\begin{figure}[!htbp]
\centering
\includegraphics[width=0.84\columnwidth]{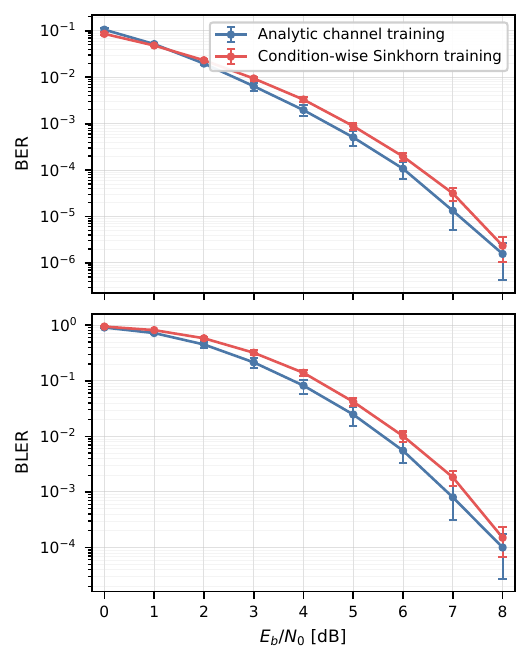}
\caption{\textbf{Long-block TurboAE AWGN check.} CNN TurboAE models at block length \(64\), rate \(1/2\), are trained through analytic AWGN or the condition-wise Sinkhorn AWGN surrogate and evaluated on analytic AWGN over \(30\) seeds.}
\label{fig:turboae-long-block-curve}
\end{figure}

\section{Discussion}
\subsection{Accuracy-latency operating point}
Condition-wise Sinkhorn drifting targets a low-latency conditional-surrogate operating point.
Diffusion remains the strongest learned reference on several hard downstream SER curves and pays for that fidelity through iterative sampling.
The one-shot drifting path uses a single generator evaluation and preserves the fixed-input structure of the channel law.
This makes it relevant when the learned channel is embedded in an inner training loop, where surrogate calls can dominate the computational budget.
Diffusion is preferable when final standalone fidelity dominates and its sampler cost is acceptable.
One-shot Sinkhorn is most attractive when repeated channel calls dominate and the observed SER gap is within tolerance.
Channel surrogates should therefore be judged by conditional-law quality and downstream communication metrics in addition to global sample-cloud agreement.

\subsection{Metric-space dependence}
The W-Flow run reinforces the metric-dependence of channel-surrogate evaluation.
Joint Sinkhorn and condition-wise Sinkhorn can swap order depending on whether the metric is global SWD, anchor-conditioned SWD, or downstream BER/SER\@.
We therefore report direct-output SWD, anchor-conditioned moment metrics, and downstream coding metrics side by side.
Hyperparameter selection from a single global SWD score would miss these differences.

\subsection{Conditional Sinkhorn scope}
The condition-wise estimator is directly applicable to analytic or simulator channels because repeated outputs can be drawn at the same transmitted symbol.
Measured channels are harder.
With only one observation per \(x\), the exact conditional law \(p(\cdot\mid x)\) cannot be observed at a point.
For that regime, the local conditional Sinkhorn approximation in Sec.~II estimates the target fixed-input output law from kernel-weighted neighborhoods in condition space and still solves the transport problem only in the output coordinate.
This keeps the fixed-input geometry intact, but introduces bandwidth selection and neighborhood bias.
The present experiments do not evaluate this measured-data approximation.
Validating it for measured-channel data is left for future work.
Practical measured-channel simulators may also need transport features that encode phase, amplitude, input-output geometry, channel-state information, or decoder-relevant projections.

\section{Conclusion}
This paper makes one-shot drifting applicable to learned channel simulation by replacing unconditional transport with condition-wise Sinkhorn transport over fixed-input conditional output laws.
The resulting generator preserves the transmitted symbol \(x\), transports only the conditional output law, and keeps the one-shot inference path that motivates drifting for inner-loop channel use.
The population formulation identifies equality of \(p(\cdot\mid x)\) and \(q_\theta(\cdot\mid x)\) for \(\mu\)-almost every \(x\) as the correct equilibrium geometry.
The implemented method is a finite-sample barycentric approximation followed by projected detached-target neural training.

Experiments on AWGN, Rayleigh, SSPA, and compact TDL channels show a clear accuracy-latency tradeoff.
Within the matched Sinkhorn drift-field ablation, condition-wise Sinkhorn is the strongest evaluated one-shot surrogate under downstream symbolic-coding checks, while diffusion remains strongest on the hardest learned-reference SER curves.
This positions the method as a low-latency, condition-preserving surrogate for settings where repeated channel calls are expensive.
The disagreement between global SWD, anchor-conditioned metrics, and downstream coding metrics makes metric choice a central part of learned channel simulation and motivates future work on channel-specific transport features and local conditional Sinkhorn approximations for measured channels.
\bibliographystyle{IEEEtran}
\bibliography{references}

\end{document}